\documentclass[11pt,a4paper]{article}

\usepackage{graphicx}
\usepackage[OT1]{fontenc}
\usepackage[latin1]{inputenc}
\usepackage[english]{babel}
\usepackage{fullpage}

\usepackage{amsmath}
\usepackage{amsthm}
\newtheorem{theorem}{Theorem}
\newtheorem{lemma}{Lemma}

\newtheorem{prop}{Proposition}

\newtheorem{definition}{Definition}
\newtheorem{corollary}{Corollary}
\usepackage{float}
\usepackage{color}

\newcommand{\todo}[1]{}
\renewcommand{\todo}[1]{{\bf{ \color{red} TODO: {#1}}}}

\newcommand{\vecteur}[1]{\ensuremath{\mathbf{#1}}}

\newcommand{\precedent}{\texttt{prec}}

\usepackage{algorithm}
\usepackage{algorithmic}
\usepackage{multirow}

\usepackage{amsfonts}
\usepackage{amssymb}
\usepackage{algorithm,algorithmic}

\usepackage[numbers]{natbib}

\begin{document}

\title{Throughput Maximization in Multiprocessor Speed-Scaling}

\author{
   Eric Angel\footnote{IBISC, Universit\'e d'Evry Val d'Essonne, France.}\\
   \and
   Evripidis Bampis
   \footnote{Sorbonne Universit\'es, UPMC Univ Paris 06, UMR 7606, LIP6, F-75005, Paris, France.}\\
   \and
   Vincent Chau$^*$\\
   \and
   Nguyen Kim Thang$^*$
}

\date{}
\maketitle
\begin{abstract}
We are given a set of $n$ jobs that have to be executed on a set of $m$ speed-scalable machines that can vary their speeds dynamically using the energy model introduced in [Yao et al., FOCS'95]. Every job $j$ is characterized by its release date $r_j$, its deadline $d_j$, its processing volume $p_{i,j}$ if $j$ is executed on machine $i$ and its weight $w_j$. We are also given a budget of energy $E$ and our objective is to maximize the weighted throughput, i.e. the total weight of jobs that are completed between their respective release dates and deadlines. We propose a polynomial-time approximation algorithm where the preemption of the jobs is allowed but not their migration. Our algorithm uses a primal-dual approach on a linearized version of a convex program with linear constraints. Furthermore, we present two optimal algorithms for the non-preemptive case where the number of machines is bounded by a fixed constant. More specifically, we consider: {\em (a)} the case of identical processing volumes, i.e. $p_{i,j}=p$ for every $i$ and $j$, for which we present a polynomial-time algorithm for the unweighted version, which becomes a pseudopolynomial-time algorithm for the weighted throughput version, and {\em (b)} the case of agreeable instances, i.e. for which $r_i \le r_j$ if and only if $d_i \le d_j$, for which we present a pseudopolynomial-time algorithm. Both algorithms are based on a discretization of the problem and the use of dynamic programming.
\end{abstract}

\section{Introduction}

Power management has become a major issue in our days. One of the mechanisms  used for saving energy in computing systems is speed-scaling where the speed of the machines can dynamically change over time. We adopt the model first introduced by 
Yao et al. \cite{YDS95} and we study the multiprocessor scheduling problem of maximizing the throughput of jobs for a given budget of energy. 
Maximizing throughput, i.e. the number of jobs or the total weight of jobs executed on time for a given budget of energy is a very natural objective in this setting. Indeed mobile devices, such as mobile phones or computers, have a limited energy capacity depending on the quality of their battery, and throughput is one of the most popular objectives in scheduling literature for evaluating the performance of scheduling algorithms for problems involving jobs that are subject to release dates and deadlines \cite{Brucker:2010:SA:1951614, Lawler90, DBLP:journals/orl/Baptiste99}. Different variants of the throughput maximization problem in the online speed-scaling setting have been
studied in the literature \cite{CCLLMW07,Li11,BCLL08,CLL10}. However, in the off-line context, only recently, an optimal pseudopolynomial-time algorithm has been proposed for the {\em preemptive}\footnote{The execution of a job may be interrupted and resumed later.} single-machine case \cite{ABC14}. Up to our knowledge no results are known for the throughput maximization problem in the multiprocessor case. In this paper, we address this issue. More specifically, we first consider the case of a set of unrelated machines and we propose a polynomial-time constant-approximation algorithm for the problem of maximizing the weighted throughput in the {\em preemptive non-migratory}\footnote{This means that the execution of a job may be interrupted and resumed later, but only on the same machine on which it has been started.} case. Our algorithm is based on the primal-dual scheme and it is inspired by the approach used in \cite{DevanurJain12:Online-matching} for the online matching problem. In the second part of the paper, we propose exact algorithms for a fixed number of identical parallel machines for instances where the  processing volumes of the jobs are all equal, or agreeable instances. Much attention has been paid to these types of instances in the speed-scaling literature (witness for instance \cite{DBLP:conf/spaa/AlbersMS07}). Our algorithms, in this part, are for the non-preemptive case and they are based on a discretization of the problem and the use of dynamic programming. 

\paragraph{Problem Definition and Notations} In the first part of the paper, we consider the problem for a set of unrelated parallel machines.
Formally, there are $m$ unrelated machines and $n$ jobs. Each job $j$ has its release date $r_{j}$, deadline $d_{j}$,
weight $w_{j}$ and its processing volume $p_{ij}$ if $j$ is assigned to machine $i$. If a job is executed on 
machine $i$ then it must be entirely processed during time interval $[r_{j},d_{j}]$ on that machine without migration. 
The \emph{weighted throughput} of a schedule is the total weight of completed jobs.  
At any time, a machine can choose a speed to process a job. If the speed of machine $i$ at time $t$ is 
$s_{i}(t)$ then the energy power at $t$ is $P(s_{i}(t))$ where $P$ is a given convex function.
Typically, one has $P(z) := z^{\alpha}$ where $2 \leq \alpha \leq 3$. The \emph{consumed energy} 
on machine $i$ is $\int_{0}^{\infty} P(s_{i}(t))dt$. 
Our objective is to maximize the weighted throughput for a given budget of energy $E$.
Hence, the scheduler has to decide the set of jobs which will be executed, assign the 
jobs to machines and choose appropriate speeds to schedule such jobs without exceeding 
the energy budget.  
In the second part of the paper we consider identical parallel machines (where the processing volume is not machine-dependent) and two  families of instances: {\em (a)} instances with identical processing volumes, i.e. $p_{i,j}=p$ for every $i$ and $j$, and {\em (b)} agreeable instances, i.e. for which $r_i \le r_j$ if and only if $d_i \le d_j$.

In the sequel, we need the following definition: Given an arbitrary convex function $P$ as the energy power function, define 
$\Gamma_{P} := \max_{z > 0} zP'(z)/P(z)$. As said before, for the most studied case in the literature one has $P(z) = z^{\alpha}$, and therefore 
$\Gamma_{P} = \alpha$.

\subsection{Our approach and contributions}
In this paper, we propose an approximation algorithm for the preemptive non-migratory
weighted throughput problem
on a set of unrelated speed-scalable machines in Section~\ref{sec:approx}. 
Instead of studying the problem directly, 
we study the related problem of minimizing the consumed energy 
under the constraint that the total weighted throughput must be at least some given 
throughput demand $W$.

For the problem of minimizing the energy's consumption under throughput constraint,
we present a polynomial time algorithm which has the following property: 
the consumed energy
of the algorithm given a throughput demand $W$ is at most that 
of an optimal schedule with throughput demand $2(\Gamma_{P}+1)W$.  
The algorithm is based on a primal-dual scheme for mathematical programs with 
linear constraints and a convex objective function. Specifically, our approach consists in 
considering a relaxation with convex objective and linear constraints. Then, we
linearize the convex objective function and construct a dual program. Using this procedure,
the strong duality is not necessarily ensured but the weak duality always holds and 
that is indeed the property that we need for our approximation algorithm. The linearization and 
the dual construction follow the scheme introduced in \cite{DevanurJain12:Online-matching}
for online matching.  

For the problem of maximizing 
the throughput under a given budget of energy, we apply a dichotomy search 
using as subroutine the algorithm for the problem of minimizing the energy's consumption for a given weighted throughput demand.
Our algorithm is a $2(\Gamma_{P}+1)$-approximation for the weighted throughput  
and the consumed energy is at most $(1+\epsilon)$ factor of the given energy budget
where $\epsilon > 0$ is an arbitrarily small constant. The violation of  the 
energy budget by a factor of $(1+\epsilon)$ is due to the arithmetic precision in computation.
The energy budget is rational as the input size is finite while the consumed 
energy for a given throughput demand could be an irrational number. For that reason,
the algorithm's running time is polynomial in the input size of the problem 
and  $1/\epsilon$. Clearly, one may be interested in finding a 
tradeoff between the precision and the running time of the algorithm.

In Section~\ref{sec:exact}, we propose exact algorithms for the non-preemptive scheduling on 
a fixed number of speed-scalable identical machines. By identical machines, we mean that $p_{i,j}=p_j$, 
i.e. the processing volume of every job is independent of the machine on which it will be executed.
We show that for the special case of the problem in which there is a single machine 
and the release dates and deadlines of the jobs are \emph{agreeable} (for every jobs $j$ and $j'$, if $r_{j} < r_{j'}$ then $d_{j} \leq d_{j'}$) the weighted throughput problem is already weakly 
$\mathcal{NP}$-hard when all the processing volumes are equal. 
We consider the following two cases
 (1) jobs have the same 
processing volume but have arbitrary release dates and deadlines; 
and (2) jobs have arbitrary processing volumes, but their release dates and deadlines are agreeable. We present pseudo-polynomial time algorithms 
based on dynamic programming for these variants. Specifically, when all jobs have the same 
processing volume, our algorithm has running time $O(n^{12m+7}W^2)$ where $W = \sum_{j} w_{j}$.  
Note that when jobs have unit weight, the algorithm has polynomial running time.
When jobs are agreeable, our algorithm has running time ${O(n^{2m+2}V^{2m+1} Wm)}$
where $V = \sum_{j} p_{j}$. Using standard techniques, these algorithms may lead to 
approximation schemes.

\subsection{Related work}
A series of papers appeared for some online variants of throughput maximization: the first work that considered throughput maximization and speed scaling in the online setting  has been presented by Chan et al.
\cite{CCLLMW07}. They considered the single machine case with release dates and deadlines and they assumed that there is an upper bound on the machine's speed.
They are interested in maximizing the throughput, and minimizing the energy among
all the schedules of maximum throughput.
They presented an algorithm which is $O(1)$-competitive with respect to both objectives.
Li \cite{Li11} has also considered the maximum throughput when
there is an upper bound in the machine's speed and he proposed
a 3-approximation greedy algorithm for the throughput and a 
constant approximation ratio for the energy consumption.
In \cite{BCLL08}, Bansal et al. improved the results of \cite{CCLLMW07}, while in \cite{LLTW07}, Lam et al. studied the 2-machines
environment.
In \cite{CLMW07}, Chan et al.  defined the energy efficiency of a schedule 
to be the total amount of work completed in time divided by the total energy usage.
Given an efficiency threshold, they considered the problem of finding a schedule of maximum throughput.
They showed that no deterministic algorithm can have competitive ratio less than the ratio of the maximum to 
the minimum jobs' processing volume.
However, by decreasing the energy efficiency of the online algorithm the competitive ratio
of the problem becomes constant.
 Finally,
in \cite{CLL10}, Chan et al. studied the problem of minimizing the energy plus a rejection penalty. The rejection penalty is a cost incurred for each job which is not completed on time and each job is associated with a value which is its importance. The authors proposed an $O(1)$-competitive algorithm for the case where the speed is unbounded and they showed that no $O(1)$-competitive algorithm exists for the case where the speed is bounded.
In what follows, we focus on the offline case.
Angel et al. \cite{ABCL13} were the first to consider the
throughput maximization problem in this setting. 
They provided
a polynomial time algorithm to solve optimally the single-machine problem for agreeable instances. More recently in \cite{ABC14}, they proved that there is a pseudo-polynomial time algorithm for solving optimally the preemptive single-machine problem with arbitrary release dates and deadlines and arbitrary processing volume. For the weighted version, the problem is $\mathcal{NP}$-hard even for instances in which all the jobs have common release dates and deadlines. Angel et al. \cite{ABCL13} showed that the problem admits a pseudo-polynomial time algorithm for agreeable instances. 
Furthermore, Antoniadis et al. \cite{AHOV13} considered a generalization of the classical knapsack problem where the objective is to maximize the total profit of the chosen items minus the cost incurred by their total weight. The case where the cost functions are convex can be translated in terms of a weighted throughput problem where the objective is to select the most profitable set of jobs taking into account the energy costs. Antoniadis et al. presented a FPTAS and a fast 2-approximation algorithm for the non-preemptive problem where the jobs have no release dates or deadlines.

Up to the best of our knowledge, no works are known for the offline throughput maximization problem in the case of multiple machines. However, many papers consider the closely related problem of minimizing the consumed energy.

For the preemptive single-machine case, Yao et al.\cite{YDS95} in their seminal paper proposed an optimal polynomial-time algorithm. Since then, a lot of papers appears in the literature (see \cite{Albers10}).
Antoniadis and Huang \cite{AH13} have considered the non-preemptive energy minimization problem. They proved that the non-preemptive single-machine case
is strongly NP-hard even for laminar instances \footnote{In a laminar instance  for any
pair of jobs $J_i$ and $J_j$, either $[r_j,d_j]\subseteq[r_i,d_i]$, 
$[r_i,d_i]\subseteq[r_j,d_j]$, or $[r_i,d_i]\cap [r_j,d_j]=\emptyset$.} and they proposed a 
$2^{5\alpha -4}$-approximation algorithm. This result has been improved recently in \cite{DBLP:conf/fsttcs/BKLLS13} where the authors proposed a $2^{\alpha-1}(1+\varepsilon)\tilde{B}_{\alpha}$-approximation algorithm, where $\tilde{B}_{\alpha}$ is the generalized Bell number.
For instances in which all the jobs have the same processing volume,
Bampis et al. \cite{BKLLN13} gave a $2^{\alpha}$-approximation for the single-machine case.
However the complexity status of this problem remained open.
In this paper, we settle this question even for the identical machine case where the number of the machine is a fixed constant. Notice that independently, Huang et al. in \cite{HO13} proposed a polynomial-time algorithm for the single machine case.

The multiple machine case where the preemption and the migration of jobs are allowed can be solved in polynomial time in
\cite{DBLP:conf/spaa/AlbersAG11}, \cite{DBLP:conf/europar/AngelBKL12} and \cite{DBLP:conf/isaac/BampisLL12}.
Albers et al. \cite{DBLP:conf/spaa/AlbersMS07} considered the multiple machine problem
where the preemption of jobs is allowed but not their migration. They showed that the problem is polynomial-time solvable for agreeable instances when the jobs
have the same processing volumes. They have also showed that it becomes strongly
NP-hard for general instances even for jobs with equal processing volumes and for this case
they proposed an $(\alpha^{\alpha}2^{4\alpha})$-approximation algorithm. For the case
where the jobs have arbitrary processing volumes, they showed that  the problem is NP-hard
even for instances with common release dates and common deadlines. Albers et
al. proposed a $2(2-1/m)^{\alpha}$-approximation algorithm for instances with common
release dates, or common deadlines, and an
$(\alpha^{\alpha}2^{4\alpha})$-approximation algorithm for
instances with agreeable deadlines. 
Greiner et al.  \cite{DBLP:conf/spaa/GreinerNS09} proposed a $B_{\alpha}$-approximation algorithm for general instances, 
where $B_{\alpha}$ is the $\alpha$-th Bell number.
Recently, the approximation ratio for agreeable instances has been improved to
$(2-1/m)^{\alpha -1}$ in \cite{BKLLN13}. For the non-preemptive multiple machine energy minimization problem, the only known result is a non-constant approximation algorithm presented in \cite{BKLLN13}.

\section{Approximation Algorithms for Preemptive Scheduling}		\label{sec:approx}

In Section \ref{sec:approx-energy}, we first study a related problem in which we look for an algorithm
that minimizes the consumed energy under the constraint of throughput demand. 
Then in Section \ref{sec:approx-throughput} we use that algorithm as a sub-routine to derive an algorithm for 
the problem of maximizing throughput under the energy constraint. 

\subsection{Energy Minimization with Throughput Demand Constraint}	\label{sec:approx-energy}

In the problem, there are $n$ jobs and $m$ unrelated machines. 
A job $j$ has release date $r_{j}$, deadline $d_{j}$,
weight $w_{j}$ and processing volume $p_{ij}$ if it is scheduled in machine $i$. 
Given throughput demand $W$, the scheduler needs to choose a 
subset of jobs, assign them to the machines and decide the speed 
to process these job in such a way that the total weight (throughput) of completed 
jobs is at least $W$ and the consumed energy is minimized. 
Jobs are allowed to be processed preemptively but without migration.

Let $x_{ij}$'s be variables indicating whether job $j$ 
is scheduled in machine $i$. 
Let $s_{ij}(t)$'s be the variable representing the speed
that the machine $i$ processes job $j$ at time $t$.  
The problem can be formulated as the following primal convex relaxation $(\mathcal{P})$.
  \begin{alignat}{3}
    \text{min} \quad \sum_{i} \int_{0}^{\infty} & P(s_{i}(t)) dt 	 \tag{$\mathcal{P}$}  \\
	\text{subject to} 	\qquad s_{i}(t) &= \sum_{j} s_{ij}(t)	\qquad &\forall i, t \notag \\
					     \sum_{i} x_{ij} &\leq 1 \qquad &\forall j 	\label{constr:box} \\
					     \int_{r_{j}}^{d_{j}} s_{ij}(t)dt & \geq p_{ij}x_{ij}   \qquad &\forall i, j 		\label{constr:completed} \\
					     \sum_{i} \sum_{j: j\notin S} w^{S}_{j}x_{ij} &\geq W - w(S) &\forall S \subset \{1,\ldots,n\} 	\label{constr:knapsack}	\\
					     x_{ij},  s_{ij}(t) &\geq 0 \qquad &\forall i, j, t \notag
  \end{alignat}
In the relaxation, constraints (\ref{constr:box}) ensures that a job can be chosen at most once.
Constraints (\ref{constr:completed}) guarantee that job $j$ must be completed if it is assigned to 
machine $i$. To satisfy the throughput demand constraint, we use the knapsack inequalities (\ref{constr:knapsack}) 
introduced in \cite{CarrFleischer00:Strengthening-integrality}. 
Note that in the constraints,  $S$ is a subset of jobs and $w^{S}_{j} := \min\{w_{j}, W - w(S)\}$.
Those constraints reduce significantly the integrality gap of the relaxation 
compared to the natural constraint $\sum_{ij} w_{ij}x_{ij} \geq W$.

Define function $Q(z) := P(z) - zP'(z)$. 
Consider the following a dual program $(\mathcal{D})$.
  \begin{alignat}{3}
    \text{max} \quad \sum_{S} & (W - w(S))\beta_{S} 
    					+ \sum_{i}\int_{0}^{\infty} Q(v_{i}(t)) dt - \sum_{j} \gamma_{j} \tag{$\mathcal{D}$} \\
	\text{s.t} \qquad \lambda_{ij} & \leq P'(v_{i}(t))   \qquad \forall i,j, \forall t \in [r_{j},d_{j}] 	\label{constr:dual-1}\\
	    				    	\sum_{S: j \notin S}w^{S}_{j}\beta_{S} &\leq  \gamma_{j} +  \lambda_{ij}p_{ij} 
							\qquad \forall i,j	\label{constr:dual-2}\\
					     \lambda_{ij}, \gamma_{j}, v_{i}(t) &\geq 0  \qquad \forall i,j, \forall t	\notag
  \end{alignat}
The construction of the dual $(\mathcal{D})$ is inspired by \cite{DevanurJain12:Online-matching} 
and is obtained by linearizing the convex objective of the primal. 
By that procedure the strong duality is not necessarily guaranteed but the weak duality always holds. 
Indeed we only need the weak duality for approximation algorithms. 
In fact, the dual $(\mathcal{D})$ gives a meaningful lower bound 
that we will exploit to design our approximation algorithm. 

\begin{lemma}[Weak Duality] 	\label{lem:formulation-PD}
The optimal value of the dual program $(\mathcal{D})$ is at most the optimal value of the primal program
$(\mathcal{P})$.
\end{lemma}
\begin{proof}
As $P$ is convex, for every $t$ and functions $s_{i}$ and $v_{i}$, we have
\begin{align}	\label{ineq:convex-energy+values}
P(s_{i}(t)) &\geq P(v_{i}(t)) + (s_{i}(t) - v_{i}(t)) P'(v_{i}(t)) \notag \\
	&= P'(v_{i}(t)) s_{i}(t) + Q(v_{i}(t))
\end{align}

Notice that if $v_{i}(t)$ is fixed then $P(s_{i}(t))$ has a lower bound in linear form
(since in that case  $P'(v_{i}(t))$ and $Q(v_{i}(t))$ are constants).
We use that lower bound to derive the dual.  
Fix functions $v_{i}(t)$ for every $1 \leq i \leq m$. 
Consider the following linear program and its dual
in the usual sense of linear programming.

\begin{figure}[ht]	
\centering{
\begin{minipage}[t]{0.3\linewidth}
\begin{center}

  \begin{alignat*}{3}
    \text{min} \quad \sum_{i} \int_{0}^{\infty} & P'(v_{i}(t))\sum_{j} s_{ij}(t) dt 	   \\
					     \sum_{i} x_{ij} &\leq 1 \quad &\forall j 	 \\
					     \int_{r_{j}}^{d_{j}} s_{ij}(t)dt & \geq p_{ij}x_{ij}   \quad &\forall i,j 		 \\
					     \sum_{i} \sum_{j: j\notin S} w^{S}_{j}x_{ij} &\geq W - w(S) &\forall S 	\\
					     x_{ij}, s_{ij}(t) &\geq 0 \quad &\forall i,j, t \notag \\
  \end{alignat*}
\end{center}
 \end{minipage}

\begin{minipage}[t]{0.6\linewidth}
\begin{center}
    \begin{alignat*}{3}
    \text{max} \quad \sum_{S}  &(W - w(S)) \beta_{S} 
    					 - \sum_{j} \gamma_{j} \\
					 \lambda_{ij} & \leq P'(v_{i}(t))   \quad &\forall i,j, \forall t \in [r_{j},d_{j}] \\
	    				    	\sum_{S: j \notin S}w^{S}_{j}\beta_{S} &\leq  \gamma_{j} +  \lambda_{ij}p_{ij} \quad &\forall i,j\\
					     \lambda_{ij}, \gamma_{j}, v_{i}(t) &\geq 0 \quad &\forall i,j, \forall t
  \end{alignat*}
\end{center}
\end{minipage}
}
\caption{Strong duality for LP}
\label{fig:weak-duality}
\end{figure}
By strong LP duality, the optimal value of theses primal and dual programs are equal. Denote that value with $OPT(v_{1}, \ldots, v_{m})$.  

Let $O_{\mathcal{P}}$ be the optimal value of the primal program $(\mathcal{P})$.
Hence, for every choice of $v_{i}(t)$, we have a lower bound on $O_{\mathcal{P}}$, i.e., 
$O_{\mathcal{P}} \geq OPT(v_{1}, \ldots, v_{m}) + \sum_{i} \int_{0}^{\infty} Q(v_{i}(t))$ 
by (\ref{ineq:convex-energy+values}). 
So $O_{\mathcal{P}} \geq \max_{v_{1}, \ldots, v_{m}} OPT(v_{1}, \ldots, v_{m}) + \sum_{i}\int_{0}^{\infty} Q(v_{i}(t))$
where $v_{i}(t)$'s are feasible solutions for $(\mathcal{D})$. 
The latter is the optimal value of the dual program $(\mathcal{D})$. Hence, 
the lemma follows.  
\end{proof}

The primal/dual programs $(\mathcal{P})$ and $(\mathcal{D})$ 
highlights main ideas for the algorithm. 
Intuitively, if a job $j$ is assigned to machine $i$ then
one must increase the speed of job $j$ in machine $i$ at 
$\arg \min P'(v_{i}(t))$ in order to always 
satisfy the constraint (\ref{constr:dual-1}). 
Moreover, when constraint (\ref{constr:dual-2}) becomes tight for some job $j$
and machine $i$, one could assign $j$ to $i$ in order to continue to raise some $\beta_{S}$
and increase the dual objective. The formal algorithm is given as follows.

\begin{algorithm}[H]
\begin{algorithmic}[1] 
\STATE Initially, set $s_{i}(t), s_{ij}(t), v_{i}(t)$ and $\lambda_{ij}, \gamma_{j}$ equal to 0
	for every job $j$, machine $i$ and time $t$.
\STATE Initially, $\mathcal{T} \gets \emptyset$.
\WHILE{$W > w(\mathcal{T})$}
	\FOR{every job $j \notin \mathcal{T}$ and every machine $i$}
		\STATE Continuously increase $s_{ij}(t)$ at $\arg \min P'(v_{i}(t))$  for $r_{j} \leq t \leq d_{j}$
			and simultaneously update $v_{i}(t) \gets v_{i}(t) + s_{ij}(t)$
			until $\int_{r_{j}}^{d_{j}}s_{ij}dt = p_{ij}$.
		\STATE Set $\lambda_{ij} \gets \min_{r_{j} \leq t \leq d_{j}} P'(v_{i}(t))$.
		\STATE Reset $v_{i}(t)$ as before, i.e., $v_{i}(t) \gets v_{i}(t) - s_{ij}(t)$ for every $t \in [r_{j},d_{j}]$.
	\ENDFOR 
	\STATE Continuously increase $\beta_{\mathcal{T}}$ until
		$\sum_{S: j \notin S} w^{S}_{j}\beta_{S} = p_{ij}\lambda_{ij}$ for some job $j$ 
		and machine $i$.
	\STATE Assign job $j$ to machine $i$. Set $s_{i}(t) \gets s_{i}(t) + s_{ij}(t)$ and $v_{i}(t) \gets s_{i}(t)$
		for every $t$.
	\STATE Set $\mathcal{T} \gets \mathcal{T} \cup \{j\}$. Moreover, set $\gamma_{j} \gets p_{ij}\lambda_{ij}$
	\STATE Reset $\lambda_{i'j'} \gets 0$ and $s_{i'j'}(t) \gets 0$ for every $(i',j') \neq (i,j)$.
\ENDWHILE
\end{algorithmic}
\caption{Minimizing the consumed energy under the throughput constraint}
\label{algo:energy}
\end{algorithm}

In the algorithm $\arg \min P'(v_{i}(t))$ for $r_{j} \leq t \leq d_{j}$ is defined as $ \{ t \: : \: t\in [r_{j},d_{j}] \mbox{ and } P'(v_i(t)) = \min_{r_{j} \leq x \leq d_{j}} P'(v_i(x)) \}$, this is usually a
set of intervals, and thus the speed $s_{ij}$ is increased simultaneously on a set of intervals.
Notice also that since $P$ is a convex function,  $P'$ is non decreasing. Hence, in line 5 of the algorithm,
$\arg \min P'(v_{i}(t))$ can be replaced by $\arg \min v_{i}(t)$; so we can avoid the computation of 
the derivative P'(z).
Given the assignment of jobs and the speed function $s_i(t)$ of each machine $i$ returned by the algorithm, in order to obtain a feasible schedule it is sufficient to schedule on each machine the jobs with the earliest deadline first order.
Note that in the end of the algorithm variables $v_{i}(t)$ is indeed equal to 
$s_{i}(t)$ --- the speed of machine $i$ for every $i$. 

The algorithm is illustrated by an exemple given in the appendix. 

\begin{lemma}		\label{lem:approx-feasible}
The solution $\beta_{S}, \gamma_{j}$ and $v_{i}(t)$ for every $i,j,S,t$ constructed 
by Algorithm~\ref{algo:energy} is feasible for the dual $(\mathcal{D})$. 
\end{lemma}
\begin{proof}
By the algorithm, variables $\lambda_{ij}$'s and variables $v_{i}(t)$'s are maintained 
in such a way that the constraints (\ref{constr:dual-1}) are always satisfied. 
Moreover, by the construction of variables 
$\beta_{S}$'s, $\lambda_{ij}$'s and $\gamma_{j}$'s, the constraints
(\ref{constr:dual-2}) are ensured (for every machine and every job).  
\end{proof}

\begin{theorem}		\label{thm:approx-main}
The consumed energy of the schedule returned by the algorithm with a throughput demand of
$W$ is at most the energy of the optimal schedule with a throughput 
demand $(2\Gamma_{P}+2)W$.
\end{theorem}
\begin{proof}
Let $OPT((2\Gamma_{P}+2)W)$ be the energy consumed by the optimal schedule
with the throughput demand $(2\Gamma_{P}+2)W$. By Lemma~\ref{lem:formulation-PD}, 
we have that 
\begin{align*}
OPT((2\Gamma_{P}+2)W) \geq 
 \sum_{S}  ((2\Gamma_{P}+2)W - w(S))\beta_{S} 
+ \sum_{i}\int_{0}^{\infty} Q(v_{i}(t)) dt - \sum_{j} \gamma_{j}
\end{align*}
where the variables $\beta_{S},v_{i},\gamma_{j}$ satisfy the same constraints in the dual
$(\mathcal{D})$. Therefore, it is sufficient to prove that latter quantity is larger than
the consumed energy of the schedule returned by the algorithm with the throughput demand $W$, denoted by $ALG(W)$.
Specifically, we will prove a stronger claim. For $\beta_{S}, \gamma_{j}$ and $v_{i}$ (which is equal to $s_{i}$) in the 
feasible dual solution constructed by Algorithm~\ref{algo:energy} with the throughput demand $W$, it always holds that
\begin{align*}
(2\Gamma_{P}+2) &\sum_{S}(W - w(S))\beta_{S} + \sum_{i}\int_{0}^{\infty} Q(s_{i}(t)) dt - \sum_{j} \gamma_{j} \geq \sum_{i} \int_{0}^{\infty} P(s_{i}(t)) dt.
\end{align*}

By the algorithm, we have that %
\begin{align}	\label{eq:approx-1}
	\sum_{i,j} p_{ij}\lambda_{ij} = \sum_{i,j: j \in \mathcal{T}} p_{ij}\lambda_{ij}
		= \sum_{j \in \mathcal{T}} \sum_{S: j\notin S} w^{S}_{j}\beta_{S} 
		= \sum_{S} \beta_{S} \biggl( \sum_{j \notin S, j \in \mathcal{T}} w^{S}_{j}\biggl)
		\leq 2 \sum_{S} \beta_{S} (W - w(S))
\end{align}
By the algorithm in the first sum $\sum_{i,j} p_{ij}\lambda_{ij}$, each term $p_{ij}\lambda_{ij} \neq 0$ iff
$j \in \mathcal{T}$ and $j$ is assigned to $i$. In the third sum, $\beta_{S} \neq 0$ iff $S$ equals
$\mathcal{T}$ at some step during the execution of the algorithm. Thus, we consider only such sets 
in that sum. Let $j^{*}$ be the last element 
added to $\mathcal{T}$. For $S \subset \mathcal{T} \setminus \{j^{*}\}$
and $\beta_{S} > 0$, by the while loop condition 
$w(S) + \sum_{j \notin S, j \in \mathcal{T} \setminus \{j^{*}\}} w^{S}_{j} < W$. Moreover,
$w^{S}_{j^{*}} \leq w^{\mathcal{T}}_{j^{*}} \leq W - w(\mathcal{T}) \leq W - w(S)$.
Hence, $\sum_{j \notin S, j \in \mathcal{T}}w^{S}_{j} \leq 2 (W - w(S))$ and the inequality 
(\ref{eq:approx-1}) follows.

Fix a machine $i$ and let $\{1, \ldots, k\}$ be the set of jobs assigned to machine $i$ (renaming jobs if 
necessary).
Let $u_{i1}(t), \ldots, u_{ik}(t)$ be the speed of machine $i$ at time $t$ after assigning 
jobs $1, \ldots, k$, respectively. In other words, $u_{i\ell}(t) = \sum_{j=1}^{\ell} s_{ij}(t)$ for every
$1 \leq \ell \leq k$. By the algorithm, we have 
$\lambda_{i\ell} = \min_{r_{\ell} \leq t \leq d_{\ell}} P'(u_{i\ell}(t))$ for every $1 \leq \ell \leq k$.
As every job $\ell$ is completed in machine $i$, $\int_{r_{\ell}}^{d_{\ell}} s_{i\ell}(t)dt = p_{i\ell}$.
Note that $s_{i\ell}(t) > 0$ only at $t$ in $\arg \min_{r_{\ell} \leq t \leq d_{\ell}} P'(u_{i\ell}(t))$.
Thus,
\begin{align}	\label{eq:approx-2}
\sum_{\ell=1}^{k} \lambda_{i\ell}p_{i\ell} 
	&= \sum_{\ell=1}^{k} \int_{r_{\ell}}^{d_{\ell}} s_{i\ell}(t)P'\biggl( \sum_{j=1}^{\ell} s_{ij}(t) \biggl)dt \notag\\
	&= \sum_{\ell=1}^{k} \int_{0}^{\infty} s_{i\ell}(t)P'\biggl( \sum_{j=1}^{\ell} s_{ij}(t) \biggl)dt	\notag \\
	&\geq  \sum_{\ell=1}^{k} \int_{0}^{\infty} 
			\biggl[P\biggl( \sum_{j=1}^{\ell} s_{ij}(t) \biggl) - P\biggl( \sum_{j=1}^{\ell-1} s_{ij}(t) \biggl)\biggl] dt \notag \\
	 &=  \int_{0}^{\infty} \biggl[P(u_{ik}(t)) - P(0) \biggl]dt \notag \\
	 &= \int_{0}^{\infty} P(s_{i}(t))dt
\end{align}
where in the second equality, note that $s_{i\ell}(t) = 0$ for $t \notin [r_{\ell},d_{\ell}]$;
the inequality is due to the convexity of $P$.

As inequality (\ref{eq:approx-2}) holds for every machines $i$, summing over all machines we get
$$
\sum_{i,j} p_{ij}\lambda_{ij} \geq \sum_{i}  \int_{0}^{\infty} P(s_{i}(t))dt.
$$
Together with (\ref{eq:approx-1}), we deduce that
\begin{align*}
(2 \Gamma_{P} + 2)& \sum_{S} \beta_{S} (W - w(S)) + \sum_{i}\int_{0}^{\infty} Q(s_{i}(t)) dt - \sum_{j} \gamma_{j} \\
\geq & \sum_{i,j} p_{ij}\lambda_{ij} + \Gamma_{P} \sum_{i}  \int_{0}^{\infty} P(s_{i}(t))dt 
+ \sum_{i}\int_{0}^{\infty} Q(s_{i}(t)) dt - \sum_{j} \gamma_{j}\notag \\
\geq & \sum_{i}  \int_{0}^{\infty} P(s_{i}(t))dt
= ALG(W).
\end{align*}
where the last inequality is due to the definition of $\Gamma_{P}$ 
(recall that $\Gamma_{P} = \max_{z} zP'(z)/P(z)$ for every $z$ such that $P(z) > 0$) and 
$\gamma_{j} = \sum_{i} \lambda_{ij}p_{ij}$ for every job $j$ (by the algorithm).
\end{proof}

\begin{corollary}	\label{cor:single-machine}
For single machine setting, 
the consumed energy of the schedule returned by the algorithm with a throughput demand of
$W$ is at most that of the optimal schedule with a throughput 
demand $2\Gamma_{P}W$.
\end{corollary}
\begin{proof}
For single machine setting, we can consider a relaxation similar to $(\mathcal{P})$
without constraints (\ref{constr:box}) and without machine index $i$ for all variables. 
The dual construction, the algorithm and the analysis remain the same. Observe that now 
there is no dual variable $\gamma_{j}$. By that point we can improve the factor 
from $2(\Gamma_{P}+1)$ to $2\Gamma_{P}$.
\end{proof}

Note that a special case of the single machine setting is the minimum knapsack problem.
In the latter, we are given a set of $n$ items, item $j$ has size $p_{j}$ and value $w_{j}$.
Moreover, given a demand $W$, the goal is to find a subset of items having minimum total 
size such that the total value is at least $W$. The problem corresponds to the single machine 
setting where all jobs have the same span, i.e., $[r_{j},d_{j}] = [r_{j'},d_{j'}]$ for all jobs
$j \neq j'$; item size and value correspond to job processing volume and weight, respectively;
and the energy power $P(z) = z$. Carnes and Shmoys \cite{CarnesShmoys08:Primal-Dual-Schema} 
gave a 2-approximation primal-dual
algorithm for the minimum knapsack problem. That result is a special case of 
Corollary~\ref{cor:single-machine} where $\Gamma_{P} = 1$ for linear function $P(z)$. 

\subsection{Throughput Maximization with Energy Constraint} 	\label{sec:approx-throughput}
We use the algorithm in the previous section as a sub-routine and 
make a dichotomy search in the feasible domain of the total throughput. 
The formal algorithm is given as follows.

\begin{algorithm}[htbp]
\begin{algorithmic}[1] 
\STATE Given a throughput demand $W$, denote $E(W)$ the 
	consumed energy due to Algorithm~\ref{algo:energy}. 
\STATE Initially, set $W_{0} \gets 0$ and $W_{1} \gets \sum_{j} w_{j}$ where the sum is taken over all jobs $j$.
\STATE Set $W \gets (W_{0} + W_{1})/2$.
\WHILE{$E(W) < E$ or $E(W) > (1+\epsilon) E$}
	\IF{$E(W) < E$} 
		\STATE $W_{0} \gets W$
	\ENDIF
	\IF{$E(W) > (1 + \epsilon) E$} 
		\STATE $W_{1} \gets W$
	\ENDIF
	\STATE $W \gets (W_{0} + W_{1})/2$
\ENDWHILE
\RETURN the schedule which is the solution of Algorithm~\ref{algo:energy} with throughput demand $W$.
\end{algorithmic}
\caption{Maximizing throughput under the energy constraint}
\label{algo:throughput}
\end{algorithm}

\begin{theorem}
Given an arbitrary constant $\epsilon > 0$, Algorithm~\ref{algo:throughput} is
$2(\Gamma_{P}+1)$-approximation in throughput with the consumed energy at most 
$(1+\epsilon)W$. The running time of the algorithm is 
polynomial in the size of input and $1/\epsilon$. 
\end{theorem}
\begin{proof}
Let $W^{*}$ be the optimal throughput with the energy budget $E$. 
Suppose that $W^{*} > 2(\Gamma_{P}+1)W$. By Theorem~\ref{thm:approx-main},
the consumed energy of the optimal schedule must be strictly larger than $E(W)$. However, the latter is 
at least $E$. So the consumed energy constraint is violated in the optimal schedule
(contradiction). Hence, $W^{*} \leq 2(\Gamma_{P}+1)W$. By the algorithm, the consumed energy
of the algorithm is at most $(1+\epsilon)E$. In Algorithm~\ref{algo:throughput}, 
the number of iterations in the while loop is proportional to the size of the input and $1/\epsilon$.
As Algorithm~\ref{algo:energy} is polynomial, the running time of Algorithm~\ref{algo:throughput}
is polynomial in the size of input and $1/\epsilon$. 
\end{proof}

\section{Exact Algorithms for Non-Pre\-emptive Scheduling}	\label{sec:exact}

\subsection{Preliminaries}

\paragraph{Notations} In this section, we consider schedules without preemption with 
a fixed number $m$ of identical machines. So the processing volume of a job $j$
is the same on every machine and is equal to $p_{j}$. 
Without loss of generality, we assume that all parameters of the problem such as 
release dates, deadlines and processing volumes of jobs are \emph{integer}.
We rename jobs in non-decreasing order of their deadlines, i.e. 
$d_1\le d_2 \le \ldots \le d_n$.
We denote by $r_{\min}:=\min_{1\le j \le n}r_j$ 
the minimum release date.
Define $\Omega$ as the set of release dates and deadlines (\textsc{edf}), 
i.e., $\Omega := \{r_j |  j=1,\ldots, n\} \cup \{d_j | j=1,\ldots, n\}$.
Let $J(k,a,b) :=\{ j | j\leq k \mbox{ and } a \le r_j < b \}$ 
be the set of jobs among the $k$ first ones w.r.t. the \textsc{edf} order,
whose release dates are within $a$ and $b$.
We consider \emph{time vectors} $\vecteur{a}=(a_1,a_2,\ldots, a_m) \in \mathbb{R}_{+}^{m}$ 
where each component $a_{i}$ is a time associated to the machines $i$ for $1 \leq i \leq m$.
We say that $\vecteur{a} \preceq \vecteur{b}$ if $a_{i} \leq b_{i}$ for every $1 \leq i \leq m$.
Moreover, $\vecteur{a}\prec \vecteur{b}$ if $\vecteur{a} \preceq \vecteur{b}$ and 
$\vecteur{a} \neq \vecteur{b}$. The relation $\preceq$ is a partial order over 
the time vectors. Given a vector $\vecteur{a}$, we denote by 
$a_{\min}:=\min_{1\le i \le m}a_i$.

\paragraph{Observations} We give some simple observations on non-preemptive scheduling with the objective 
of maximizing
throughput under the energy constraint. First, it is well known that
due to the convexity of the power function $P(z) := z^{\alpha}$, 
each job runs at a constant speed during its whole
execution in an optimal schedule. This follows from Jensen's Inequality.
Second, for a restricted version of the problem in which there is a single machine, jobs have the same 
processing volume and are agreeable, the problem is already $\mathcal{NP}$-hard.
That is proved by a simple reduction from {\sc Knapsack}.

\begin{prop}	\label{prop_np_hard}
The problem of maximizing the weighted throughput on the case where jobs have agreeable deadline and have the same processing volume is weakly $\mathcal{NP}$-hard.
\end{prop}
\begin{proof}
Let $\Pi$ be the the weighted throughput problem on the case where jobs have agreeable deadline and have the same processing volume.
In an instance of the {\sc Knapsack} problem we are given a set of $n$ items, 
each item $j$ has a value $\kappa_j$ and a size $c_j$.
Given a capacity $C$ and a value $K$, we are asked for a subset of items with total 
value at least $K$ and total size at most $C$.

Given an instance of the {\sc Knapsack} problem, 
construct an instance of problem $\Pi$ as follows.
For each item $j$, create a job $j$ with 
$r_j:=\sum_{\ell=1}^{j-1} c_\ell$, $d_j  := \sum_{\ell=1}^{j}c_\ell = $ $r_{j} + c_{j}$, $w_j:=\kappa_j$ and $p_j:=1$.
Moreover, we set $E:=C$, i.e. the budget of energy is equal to $C$.

We claim that the instance of the {\sc Knapsack} problem is feasible 
if and only if there is a feasible schedule for problem $\Pi$
of total weighted throughput at least $K$.

Assume that the instance of the {\sc Knapsack} is feasible.
Therefore, there exists  a subset of items $J'$ such that $\sum_{j\in J'}\kappa_j\geq K$ and $\sum_{j\in J'}c_i\leq C$.
Then we can schedule all jobs corresponding to item in $J'$ 
with constant speed equal to 1. That gives a feasible schedule
with total energy consumption at most $C$ and the total weight at least $K$. 

For the opposite direction of our claim, assume there is a feasible schedule for
problem $\Pi$ of total weighted throughput at least $K$.
Let $J'$ be the jobs which are completed on time in this schedule.
Clearly, due to the convexity of the speed-to-power function, the schedule that 
executes the jobs in $J'$ with constant speed is also feasible.
Since the latter schedule is feasible, we have that $\sum_{j\in J'}(d_j-r_j)\leq C$.
Moreover, $\sum_{j\in J'}w_j\geq K$.
Therefore, the items which correspond to the jobs in $J'$ form a feasible solution for the {\sc Knapsack}. 
\end{proof}

The hardness result rules out the possibility of polynomial-time exact algorithms for the problem. 
However, as the problem is weakly $\mathcal{NP}$-hard, there is still possibility for approximation schemes. 
In the following sections, we show pseudo-polynomial-time exact algorithms for instances with 
equal processing volume jobs 
and agreeable jobs.

\subsection{Equal Processing Volume}

In this section, we assume that $p_{j}=p$ for every job $j$.

\begin{definition}\label{def:Theta_x_y}
Let $\Theta_{a,b} :=\{a+ \ell \cdot \frac{b-a}{k} ~|~ k = 1,\ldots ,n \mbox{ and }$ 
$ \ell = 0,\ldots ,k \mbox{ and } a\le b
 \}$. Moreover, $\Theta :=\bigcup\{\Theta_{a,b} | a,b\in\Omega\}$.
\end{definition}

The following lemma gives an observation on the structure of an optimal 
schedule.

\begin{lemma}\label{theta}
There exists an optimal schedule in which the starting time and completion time
of each job belong to the set $\Theta$.
\end{lemma}
\begin{proof}
Let $\mathcal{O}$ be an optimal schedule and $\mathcal{O}_i$
be the corresponding schedule $\mathcal{O}$ on machine $i$.
$\mathcal{O}_i$ can be partitioned into successive blocks
of jobs where the blocks are separated by idle-time periods. 
Consider a block $B$ and decompose
$B$ into maximal sub-blocks $B_1,\ldots , B_{k}$ such 
that all the jobs executed inside a sub-block $B_{\ell}$ are scheduled with
the same common speed $s_{\ell}$ for $1 \leq \ell \leq k$. 
Let $j$ and ${j'}$ be two consecutive jobs such
that $j$ and ${j'}$ belong to two consecutive sub-blocks, let's say $B_{\ell}$ and 
$B_{\ell+1}$. Then either $s_{\ell}> s_{\ell+1}$ or $s_{\ell} < s_{\ell+1}$.
In the first case, the completion time of job $j$ (which is also the starting time
of job ${j'}$) is necessarily $d_j$, otherwise we could obtain a better schedule
by decreasing (resp. increasing) the speed of job $j$ (resp. ${j'}$).
For the second case, a similar argument shows that the completion time of job $j$
is necessarily $r_{j'}$. Hence, each sub-block begins and finishes at a date which belong to $\Omega$. 

Consider a sub-block $B_{\ell}$ and let $a,b$ be its starting and completion times. 
As jobs have the same volume and 
the jobs scheduled in $B_{\ell}$ are processed non-preemptively by the same speed, 
their starting and completion times must belong to $\Theta_{a,b}$.
\end{proof}

Using Lemma~\ref{theta} we can assume that each job is processed at some speed which belong to the following set.

\begin{definition}	\label{set_of_speed}
Let $\Lambda :=\{ \frac{\ell \cdot p}{b-a}~|~ \ell = 1,\ldots ,n \mbox{ and }$
$a,b\in\Omega  \mbox{ and } a< b \}$ be the set of different speeds.
\end{definition}

\begin{definition}\label{def:Eksxtu}
For $0\leq w\leq W$, define $E_k(\vecteur{a},\vecteur{b},w,e)$ as the minimum energy consumption 
of a non-preemptive (non-migration) schedule $\mathcal{S}$ such that 
\begin{itemize}
\item $S \subset J(k,a_{\min},b_{\min})$ 
	and  $\sum_{j\in S}w_j\ge w$ where $S$ is the set of jobs scheduled in $\mathcal{S}$, 
\item if $j \in S$ is assigned to machine $i$ then it is entirely processed in 
	$[a_i,$ $b_i]$ for every $1\le i \le m$,
\item $\vecteur{a} \preceq \vecteur{b}$,
\item for some machine $1 \leq h \leq m$, it is idle during interval $[a_{h},e]$,
\item for arbitrary machines $1 \leq i \neq i' \leq m$, $b_{i'}$ is at least the last 
	starting time of a job in machine $i$. 
\end{itemize}
\end{definition}
Note that $E_k(\vecteur{a},\vecteur{b},w,e) = \infty$ if no such schedule $\mathcal{S}$
exists.

\begin{prop}\label{prop_E}
One has
\begin{align*}
E_0(\vecteur{a},\vecteur{b},0,e) &=0 \notag \\
E_0(\vecteur{a},\vecteur{b},w,e) &=+ \infty~\forall w \neq 0	\notag \\
E_k(\vecteur{a},\vecteur{b},w,e) 
  	&= \min \begin{cases} 
		E_{k-1}(\vecteur{a},\vecteur{b},w,e) 	\\
		E'
	\end{cases}
\end{align*}
where
\begin{align*}
E'=\min_{
			\substack{
			\vecteur{u}\in \Theta^m\\
			\vecteur{a}\preceq \vecteur{u} \prec \vecteur{b}\\
			s \in \Lambda, 1 \leq h \leq m,\\
			e' = u_{h} + \frac{p}{s}\\
			r_{k} \le u_{h} < e' \le d_k\\
			0\le w' \le w-w_k}}
		&\left\{
		\begin{array}{c}
			E_{k-1}(\vecteur{a},\vecteur{u},w',e) \\
			+ \frac{p^\alpha}{(e'-u_{h})^{\alpha-1}}\\
 			+E_{k-1}(\vecteur{u},\vecteur{b},w-w'-w_k,e')
 		\end{array} \notag
		\right\}
\end{align*}
\end{prop}

\begin{figure}[ht]
\begin{center}
\includegraphics[width=0.7\textwidth]{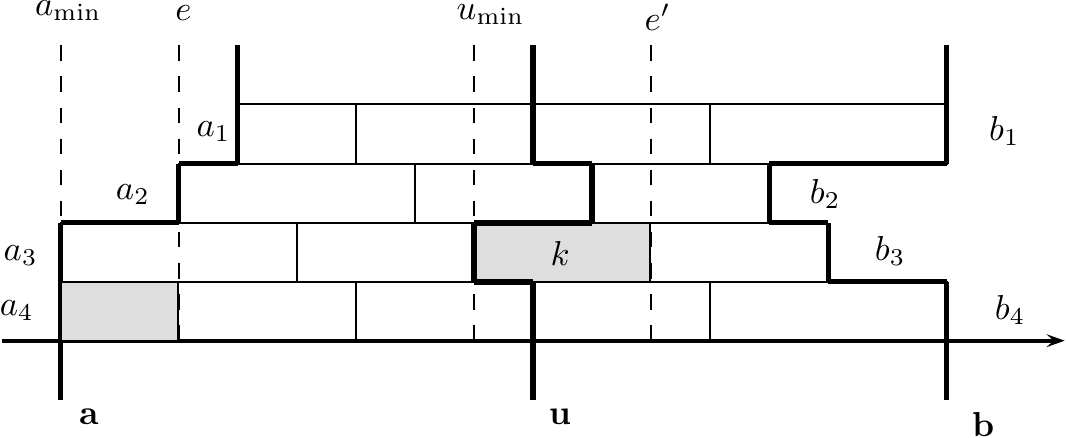}
\end{center}
\caption{Illustration of Proposition~\ref{prop_E}}
\label{fig_E}
\end{figure}

\begin{proof}
The base case for $E_{0}$ is straightforward. We will prove the recursive formula 
for $E_k(\vecteur{a},\vecteur{b},w,e)$. There are two cases: (1) either in the schedule
that realizes $E_k(\vecteur{a},\vecteur{b},w,e)$, job $k$ is not chosen, so 
$E_k(\vecteur{a},\vecteur{b},w,e) = E_{k-1}(\vecteur{a},\vecteur{b},w,e)$; 
(2) or $k$ is chosen in that schedule. In the following, we are interested 
by that second case.

\paragraph{We first prove that $E_k(\vecteur{a},\vecteur{b},w,e)\le E'$}
Fix some arbitrary time vector $\vecteur{a} \prec \vecteur{u} \prec \vecteur{b}$
and weight $0 < w' < w - w_{k}$ and time $e'$ such that $r_{k} \leq e' = u_{i} + \frac{p}{s} \le d_k$
for some $s \in \Lambda$ and some machine $h$. 
Consider a schedule $\mathcal{S}_1$ 
that realizes $E_{k-1}(\vecteur{a},\vecteur{u},w',e)$ and $\mathcal{S}_2$ 
a schedule that realizes $E_{k-1}(\vecteur{u}, \vecteur{b},w - w' - w_k, e')$.
We build a schedule with $\mathcal{S}_1$ from $\vecteur{a}$ to $\vecteur{u}$ and with 
$\mathcal{S}_2$ from $\vecteur{u}$ to $\vecteur{b}$ and job $k$ scheduled within $\mathcal{S}_2$
during $[u_{i},e']$ on machine $h$. 
Recall that by definition of $E_{k-1}(\vecteur{u}, \vecteur{b},w - w' - w_k, e')$, 
machine $h$ does not execute any job during $[u_{h},e']$.
Obviously, the subsets $J(k-1,a_{\min},u_{\min})$ and $J(k,u_{\min},b_{\min})$ do not intersect, so
this is a feasible schedule which costs at most 
$$E_{k-1}(\vecteur{a},\vecteur{u},w',e)  +
				\frac{p^\alpha}{(e'-u_{h})^{\alpha-1}}
 			+E_{k-1}(\vecteur{u},\vecteur{b},w-w'-w_k,e').
$$
As that holds for every time vector $\vecteur{a} \prec \vecteur{u} \prec \vecteur{b}$
and weight $0 < w' < w - w_{k}$ and time $e'$ such that $r_{k} \leq e' = u_{h} + \frac{p}{s} \le d_k$
for some $s \in \Lambda$ and some machine $h$, we deduce that 
$E_k(\vecteur{a},\vecteur{b},w,e) \le E'$.

\paragraph{We now prove that $E'\le E_k(\vecteur{a},\vecteur{b},w,e)$}
Let $\mathcal{S}$ be the schedule that realizes $E_k(\vecteur{a},\vecteur{b},w,e)$ in which
the starting time of job $k$ is maximal. Suppose that job $k$ is scheduled on machine $h$ and 
its starting time is denoted as $u_h$. For every machine $i \neq h$, define $u_{i} \geq u_{h}$ be the earliest 
completion time of a job which is completed after $u_{h}$ on machine $i$ by schedule $\mathcal{S}$. 
If no job is completed after $u_{h}$ on machine $i$ then define $u_{i} = b_{i}$. Hence, we have a time vector 
$\vecteur{a} \prec \vecteur{u} = (u_{1},\ldots, u_{m}) \prec \vecteur{b}$.

We split $\mathcal{S}$ into two sub-schedules $\mathcal{S}_1\subseteq \mathcal{S}$ 
and $\mathcal{S}_2 = \mathcal{S} \setminus (\mathcal{S}_{1} \cup \{k\})$
such that $j \in \mathcal{S}_1$ if it is started and completed in $[a_{i},u_{i}]$ for 
some machine $i$. Note that such job $j$ has release date $r_{j} \in [a_{\min}, u_{\min}[$.

We claim that for every job $j \in \mathcal{S}_{2}$, $r_{j} \geq u_{h}$ where $u_{h} = u_{\min}$
by the definition of vector $\vecteur{u}$. 
By contradiction, suppose that some job $j \in \mathcal{S}_{2}$
has $r_{j} \leq u_{h}$, meaning that job $j$ is available at the starting time of 
job $k$. By definition of $\mathcal{S}_{1}$ and $\mathcal{S}_{2}$, 
job $j$ is started after the starting time of job $k$. Moreover, $j < k$ means that $d_{j} \leq d_{k}$.
Thus, we can swap jobs $j$ and $k$ (without modifying the machine speeds). 
Since all jobs have the same volume, this operation 
is feasible. The new schedule has the same energy cost while the starting time of job $k$ is
strictly larger. That contradicts the definition of $\mathcal{S}$.

\begin{figure}[!h]
\begin{center}
\includegraphics[width=0.7\textwidth]{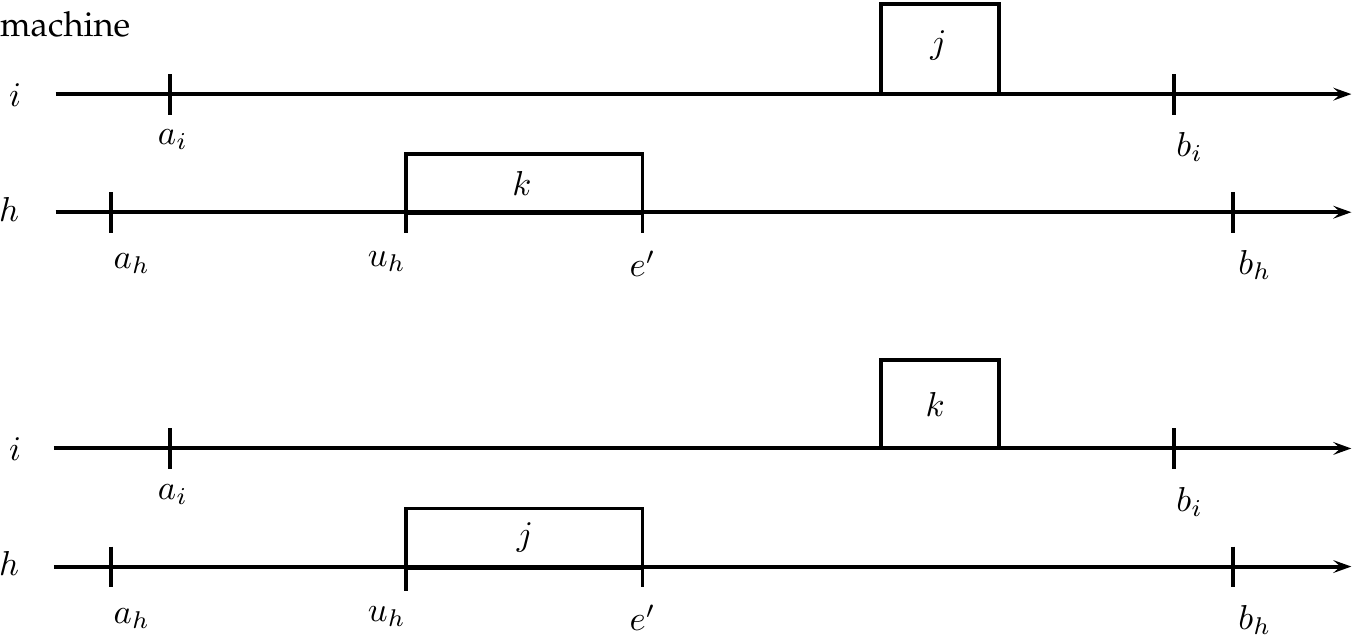}
\end{center}
\caption{Illustration of the swap argument}
\label{fig_swap}
\end{figure}

Therefore, all jobs in $\mathcal{S}_{1}$ have release dates in $[a_{\min},u_{\min}[$
and all jobs in $\mathcal{S}_{2}$ have release dates in $[u_{\min},b_{\min}[$.
Moreover, with the definition of vector $\vecteur{u}$, the schedules
$\mathcal{S}_{1}$ and $\mathcal{S}_{2}$ are valid (according to Definition~\ref{def:Eksxtu}).
Let $s_{hk} \in \Lambda$ be the speed that machine $h$ processes job
$k$ in $\mathcal{S}$.
Hence, the consumed energies by schedules $\mathcal{S}_{1}$ and 
$\mathcal{S}_{2}$ are at least $E_{k-1}(\vecteur{a},\vecteur{u},w',e)$
and $E_{k-1}(\vecteur{u},\vecteur{b},w-w_{k} - w',u_{h} + \frac{p}{s_{hk}})$ where 
$w'$ is the total weight of jobs in $\mathcal{S}_{1}$. 
We have
\begin{align*}
E_k(\vecteur{a},\vecteur{b},w,e)
	\geq ~&E_{k-1}(\vecteur{a},\vecteur{u},w',e)
		+ \frac{p^{\alpha}}{(u_{h} + \frac{p}{s_{hk}} - u_{h})^{\alpha-1}}\\
		+ &E_{k-1}(\vecteur{u},\vecteur{b},w-w_{k} - w',u_{h}+ \frac{p}{s_{hk}}) = E'.
\end{align*}
Therefore, we deduce that $E' = E_k(\vecteur{a},\vecteur{b},w,e)$ in case 
job $k$ is chosen in the schedule that realizes $E_k(\vecteur{a},\vecteur{b},w,e)$.
The proposition follows.
\end{proof}

\begin{theorem}\label{prop_complexity_non_preempt}
The dynamic program in Proposition~\ref{prop_E} has a running time of $O(n^{12m+7}W^2)$.
\end{theorem}

\begin{proof}
Denote $r_{\min} = \min_{1 \leq j \leq n} r_{j}$.
Given an energy budget $E$, the objective function is
$\max \{ w~|~E_n((r_{\min},\ldots,r_{\min}),(d_{n},\ldots,d_n),w,r_{\min}) \le E \}$.
The values of $E_k(\vecteur{a},\vecteur{b},w,e)$ are stored in a multi-dimensional array of
size $O(|\Theta|^{2m} |\Lambda| nW)$.
Each value need $O(|\Lambda| |\Theta|^{m}$ $mW)$ time to be computed thanks to Proposition~\ref{prop_E}.
Thus we have a total running time of $O(|\Theta|^{3m} |\Lambda|^2 nmW^2)$.
This leads to an overall time complexity $O(n^{12m+7}mW^2)$.
\end{proof}

\subsection{Agreeable Jobs}

In this section, we focus on another important family of instances. More precisely,
we assume that the jobs have {\em agreeable} deadlines,
i.e.  for any pair of jobs $i$ and $j$, one has $r_i\leq r_j$ if and only if
$d_i\leq d_j$.

Based on Definition~\ref{def:Theta_x_y}, we can extend the set of starting and completion times
for each job into the set $\Phi$.
\begin{definition}\label{def:Phi}
Let $\Phi_{a,b} :=\{a+ \ell \cdot \frac{b-a}{k}~|~ k = 1,\ldots ,V \mbox{ and } $
$\ell = 0,\ldots ,k
\}$ with $V:=\sum_j p_j$, and $\Phi:=\bigcup\{\Phi_{a,b}~|~a,b\in\Omega\}$.
\end{definition}

The following lemmas show the structure of an optimal schedule that we will use in order to design our algorithm.

\begin{lemma}\label{agreeable}
There exists an optimal solution in which all jobs in each machine are
scheduled according to the Earliest Deadline First ({\sc edf}) order without preemption.
\end{lemma}

\begin{proof}
Let $\mathcal{O}$ be an optimal schedule.
Let $j$ and $j'$ be two consecutive jobs that are scheduled on the same machine $i$ in $\mathcal{O}$.
We suppose that job $j$ is scheduled before job ${j'}$ with $d_{j'}\le d_j$.
Let $a$ (resp. $b$) be the starting time (resp. completion time)
of job $j'$ (resp. job $j$) in $\mathcal{O}$. Then, we have necessarily $r_{j'}\leq r_j \leq a< b \leq d_{j'}\leq d_j$. The execution of jobs $j$ and $j'$ can be swapped in the time interval 
$[a,b]$.
Thus we obtain a feasible schedule $\mathcal{O}'$ in which job ${j'}$ is scheduled before job
$j$ with the same energy consumption.
\end{proof}

\begin{lemma}\label{Phi}
There exists an optimal {\sc edf} schedule $\mathcal{O}$ in which each job in $\mathcal{O}$
has its starting time and its completion time that belong to the set $\Phi$.
\end{lemma}
\begin{proof}
We proceed as in Lemma~\ref{theta}. We partition an optimal schedule
$\mathcal{O}$ into blocks and sub-blocks where the starting and completion times of 
every sub-blocks belong to the set $\Lambda$. Consider an arbitrary sub-block. As all the 
parameters are integer, the total volume of the sub-block is also an integer in $[0,V]$ and the 
total number of jobs processed in the sub-block is bounded by the total volume. 
Thus the starting and completion times of any job in the sub-block must belong to the set $\Phi$. 
\end{proof}

By Lemma~\ref{Phi}, we can assume that each job is processed with a speed that 
belongs to the following set.

\begin{definition}\label{def:Delta}
Let $\Delta:=\{ \frac{i}{b-a}~|~ i = 1,\ldots ,V \mbox{ and } a,b\in\Omega\}$ 
be the set of different speeds.
\end{definition}

\begin{definition}\label{def:Fktw}
For $1\leq w\leq W$, define $F_k(\vecteur{b},w)$ as the minimum energy consumption 
of an non-preemptive (and a non-migratory) schedule $\mathcal{S}$ such that:
\begin{itemize}
\item $S\subseteq J(k,r_{\min},b_{\min})$ and $\sum_{j\in S}w_j\ge w$
	where $S$ is the set of jobs scheduled in $\mathcal{S}$
\item if $j \in S$ is assigned to machine $i$ then it is entirely processed in 
	$[r_{\min},$ $b_i]$ for every $1\le i \le m$.
\end{itemize}
\end{definition}
Note that $F_k(\vecteur{b},w)= \infty$ if no such schedule $\mathcal{S}$
exists.



For a vector $\vecteur{b}$ and a speed $s \in \Delta$,
let $\precedent_k(\vecteur{b},s)$ be the set of vectors
$\vecteur{a} \prec \vecteur{b}$ such that there always exists some machine 
$1 \leq h \leq m$ with the following properties:
\[
\left\{
\begin{aligned}
r_{k} \leq & a_h = ~ \min \{ b_{h}, d_{k} \} - \frac{p_k}{s},~a_{h} \in \Phi   \\
&\qquad a_i = b_i ~\forall i \neq h.
\end{aligned}
\right.
\]

\begin{prop}\label{prop_F}
One has
\begin{align*}
F_0(\vecteur{b},0)&=0\notag \\
F_0(\vecteur{b},w)&=+ \infty~\forall w\neq 0 \notag \\
F_k(\vecteur{b},w)
	&=\min \{ F_{k-1}(\vecteur{b},w), F' \}
\end{align*}
where 
\begin{align*}
F' = \min_{
			\substack{
				s\in \Delta\\
				\vecteur{a}=\precedent_k(\vecteur{b},s)	\\
			}
		}
		\left\{
		F_{k-1}(\vecteur{a},w-w_k) 
		+ p_k s^{\alpha-1} 
		\right\}
\end{align*}
\end{prop}

\begin{figure}[ht]
\begin{center}
\includegraphics[width=0.7\textwidth]{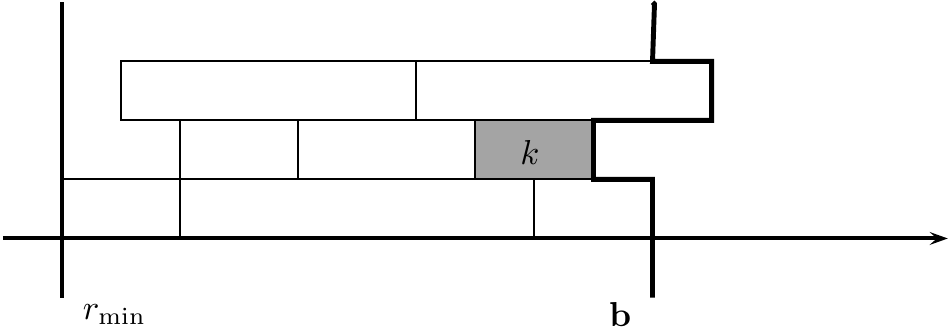}
\end{center}
\caption{Illustration of Proposition~\ref{prop_F}}
\label{fig_F}
\end{figure}

\begin{proof}

The base case for $F_{0}$ is straightforward. We will prove the recursive formula 
for $F_k(\vecteur{b},w)$. There are two cases: (1) either in the schedule
that realizes $F_k(\vecteur{b},w)$, job $k$ is not chosen, so 
$F_k(\vecteur{b},w) = F_{k-1}(\vecteur{b},w)$; 
(2) or $k$ is chosen in that schedule. In the following, we are interested 
in the case when $k$ is chosen.

\paragraph{We first prove that $F_k(\vecteur{b},w)\le F'$}
Fix some arbitrary time vector $\vecteur{a}\prec \vecteur{b}$  and $a_i$ such that
$r_{k} \leq a_i = ~ \min \{ b_{i}, d_{k} \} - \frac{p_k}{s}$ for some 
speed $s\in \Delta$ and some machine $i$.
Then we have $\vecteur{a}=(b_1,\ldots, \min \{ b_{i}, d_{k} \}-\frac{p_k}{s},\ldots,b_m)$.
Consider a schedule $\mathcal{S}$ that realizes $F_{k-1}(\vecteur{a},w-w_k)$.
We build a schedule with $\mathcal{S}$ from $(r_{min},\ldots , r_{min})$ to $\vecteur{a}$ 
and job $k$ is scheduled on machine $i$
during $[a_i,\min \{ b_{i}, d_{k} \}]$ 
and an idle period during $[\min \{ b_{i}, d_{k} \},b_{i}]$. 
So this is a feasible schedule which costs at most 
$$
F_{k-1}(\vecteur{a},w-w_k) + p_k s^{\alpha-1}
$$
As that holds for every time vector $\vecteur{a} \prec \vecteur{b}$
and some speed $s \in \Delta$ and some machine $i$, we deduce that 
$F_k(\vecteur{b},w) \le F'$.

\paragraph{We now prove that $F'\le F_k(\vecteur{b},w)$}
Let $\mathcal{S}$ be the schedule that realizes $F_k(\vecteur{b},w)$ in which
the starting time of job $k$ is maximal.
We consider the sub-schedule $\mathcal{S}'=\mathcal{S}\setminus \{k\}$.
We claim that all the jobs of $\mathcal{S}'$ 
are completed before $\vecteur{a}\in \precedent_k(\vecteur{b},s)$ which
is the vector obtained from $\vecteur{b}$ after removing job $k$.

Hence the cost of the schedule 
$\mathcal{S}'$ is at least $F_{k-1}(\vecteur{a},w-w_k)$. Thus,
$$F_k(\vecteur{b},w) \geq F_{k-1}(\vecteur{a},w-w_k) + p_k(s)^{\alpha-1}  =F'$$
Therefore, we deduce that $F'=F_k(\vecteur{b},w)$ in case job $k$ is chosen in the schedule that realizes $F_k(\vecteur{b},w)$. The proposition follows.
\end{proof}

\begin{theorem}\label{prop_complexity_agreeable}
The dynamic programming in Proposition~\ref{prop_F} has a total running time of ${O(n^{2m+2}V^{2m+1} Wm)}$.
\end{theorem}

\begin{proof}
Given an energy budget $E$, the objective function is
$\max\{w~|~F_n({\vecteur{b},w}) \le E,~1\le w\le W,~\vecteur{b}\in \Phi^m:~d_{1} \le b_i\le d_{n}~\forall i\}$.
The values of $F_k(\vecteur{b},w)$ are stored in a multi-dimensional array of
size $O(n|\Theta|^{m}W)$.
Each value need $O(|\Delta|  Wm)$ time to be computed thanks to Proposition~\ref{prop_F}.
Thus we have a total running time of $O(n|\Theta|^{m} |\Delta| W m)$.
This leads to an overall time complexity $O(n^{2m+2}V^{2m+1}$ $Wm)$.
\end{proof}

\bibliographystyle{abbrv}
\bibliography{biblio} 
\appendix
\section{Execution of Algorithm 1}

In this example, we have $m=2$ unrelated machines, $n =4$ jobs and each job have the same weight, i.e.
$w_j=1 \forall j$. We want to compute the energy's consumption when we have to choose 
$W=3$ jobs according to our algorithm.

Let $P(z)=z^\alpha$ with $\alpha=3$ be the power function of the machines.
And let the derivative function $P'(z)=3z^2$.

The processing volume of each job is given in the following table.
\begin{center}
\begin{tabular}{|c|c|c|c|c|}
\hline 
$i \backslash j$ & 1 & 2 & 3 & 4 \\ 
\hline 
1 & 1 & 3 & 4 & 2 \\ 
\hline 
2 & 2 & 5 & 3 & 1 \\ 
\hline 
\end{tabular} 
\end{center}

\begin{figure}[h]
\begin{center}
\includegraphics[width=0.6\textwidth]{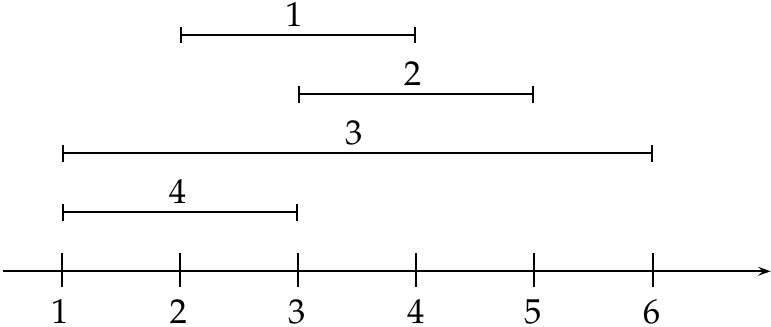}
\caption{Instance of 4 jobs with the respective release date and deadline}
\label{exemple_appendix}
\end{center}
\end{figure}

\paragraph{Step 1}
At this step, the set of chosen jobs is $T=\emptyset$\\

We continuously increase the speed $s_{ij}(t)$ for each job $j$ and each machine $i$ with $r j \le t \le d j$.
Then we obtain the value of $\lambda_{ij} \gets \min_{r_{j} \leq t \leq d_{j}} P'(v_{i}(t))$.

\begin{table}[H]
\centering
\begin{tabular}{|c|c|c|c|c|}
\hline 
$i \backslash j$ & 1 & 2 & 3 & 4 \\ 
\hline 
1 & $P'(\frac{1}{2})=\frac{3}{4}$ & $P'(\frac{3}{2})=\frac{27}{4}$ & $P'(\frac{4}{5})=\frac{48}{25}$ & $P'(1)=3$ \\ 
\hline 
2 & $P'(1)=3$ & $P'(\frac{5}{2})=\frac{75}{4}$ & $P'(\frac{3}{5})=\frac{27}{25}$ & $P'(\frac{1}{2})=\frac{3}{4}$ \\ 
\hline 
\end{tabular} 
\caption{Table of $\lambda_{ij}$ at Step 1}
\end{table}

\begin{table}[H]
\centering
\begin{tabular}{|c|c|c|c|c|}
\hline 
$i \backslash j$ & 1 & 2 & 3 & 4 \\ 
\hline 
1 & 3/4 & 81/4 & 192/25 & 6 \\ 
\hline 
2 & 6 & 625/4 & 81/25 & 3/4 \\ 
\hline 
\end{tabular} 
\caption{Table of $\lambda_{ij}p_{ij}$ at Step 1}
\end{table}

We continuously increase $\beta_{\mathcal{T}}$ until
		$\sum_{S: j \notin S} w^{S}_{j}\beta_{S} = p_{ij}\lambda_{ij}$ for some job $j$ 
		and machine $i$.

Since $\beta_{\mathcal{S}}=0~\forall \mathcal{S}$ at this step and we can only modify 
the value of $\beta_{\mathcal{T}}=\beta_{\emptyset}$, then we have to find the maximum 
value of $\beta_{\emptyset}$ such that one of the constraint becomes tight.

$w^{\emptyset}_{j} \beta_{\emptyset}=\min \{ p_{ij}\lambda_{ij}   \} =  \frac{3}{4}$ and 
$\gamma_1=\frac{3}{4}$

Thus Job 1 is affected to machine 1 and $T=\{ 1 \}$

\begin{figure}[H]
\begin{center}
\includegraphics[width=0.7\textwidth]{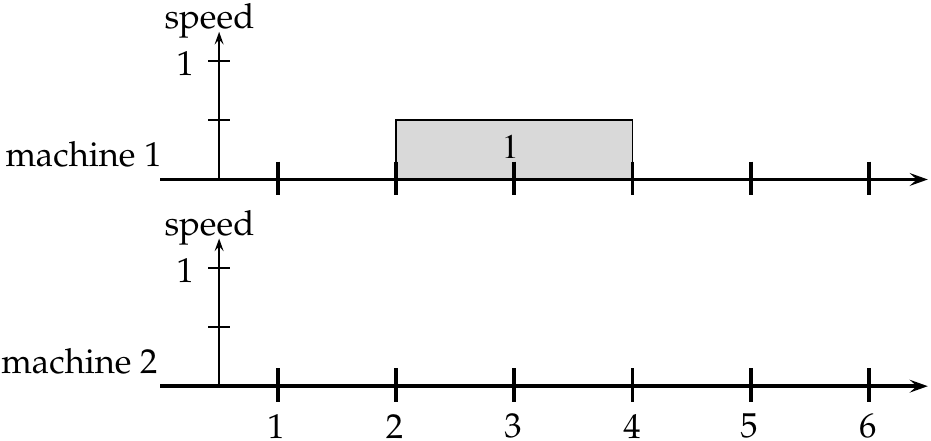} 
\end{center}
\caption{Speed profile $v_i(t)$ at the end of Step 1}
\label{speed_profile_step1}
\end{figure}

\paragraph{Step 2}

At this step, the set of chosen jobs is $T=\{ 1  \}$ and the speed profile $v_i(t)$ can be found in Figure~\ref{speed_profile_step1}

\begin{table}[H]
\centering
\begin{tabular}{|c|c|c|c|}
\hline 
$i \backslash j$  & 2 & 3 & 4 \\ 
\hline 
1  & $P'(\frac{7}{4})=\frac{147}{16}$ & $P'(1)=3$ & $P'(\frac{5}{4})=\frac{75}{16}$ \\ 
\hline 
2 & $P'(\frac{5}{2})=\frac{75}{4}$ & $P'(\frac{3}{5})=\frac{27}{25}$ & $P'(\frac{1}{2})=\frac{3}{4}$ \\ 
\hline 
\end{tabular} 
\caption{Table of $\lambda_{ij}$ at Step 2}
\end{table}

\begin{table}[H]
\centering
\begin{tabular}{|c|c|c|c|}
\hline 
$i \backslash j$ & 2 & 3 & 4 \\ 
\hline 
1 & 441/16 & 12 & 150/16 \\ 
\hline 
2  & 625/4 & 81/25 & 3/4 \\ 
\hline 
\end{tabular} 
\caption{Table of $\lambda_{ij}p_{ij}$ at Step 2}
\end{table}

At this step we have only $\beta_{\emptyset}$ which is positive.
Then we have 
$\beta_{ \{ 1  \}}=0$

\begin{align*}
w^{\emptyset}_{j} \beta_{\emptyset}+w^{\{ 1 \}}_{j} \beta_{\{ 1 \}}&=\min \{ p_{ij}\lambda_{ij}   \} \\
\frac{3}{4}+\beta_{\{ 1 \}}&=\min \{ p_{ij}\lambda_{ij}   \}\\
\beta_{\{ 1 \}}&=0
\end{align*}
Job 4 is affected to machine 2, $\gamma_4=\frac{3}{4}$ and $T=\{ 1,4 \}$.

\begin{figure}[H]
\begin{center}
\includegraphics[width=0.7\textwidth]{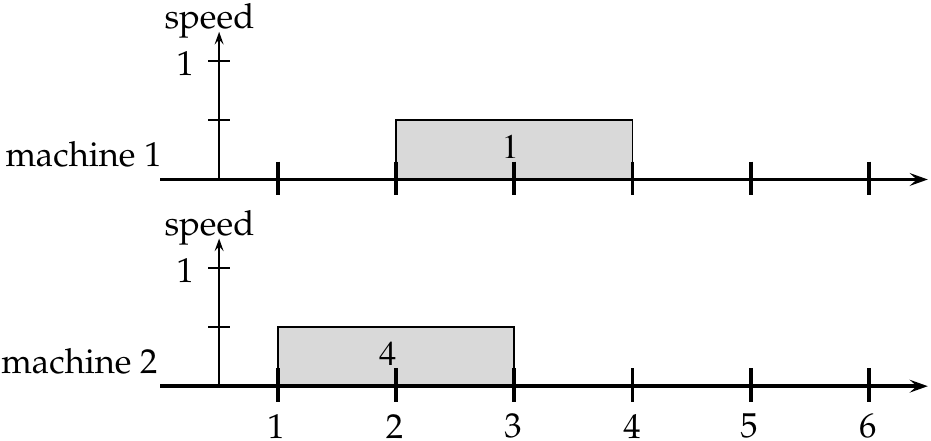} 
\end{center}
\caption{Speed profile $v_i(t)$ at the end of Step 2}
\label{speed_profile_step2}
\end{figure}

\paragraph{Step 3}

$T=\{ 1,4  \}$\\

\begin{table}[H]
\centering
\begin{tabular}{|c|c|c|}
\hline 
$i \backslash j$  & 2 & 3  \\ 
\hline 
1  & $P'(\frac{7}{4})=\frac{147}{16}$ & $P'(1)=3$  \\ 
\hline 
2 & $P'(\frac{5}{2})=\frac{75}{4}$ & $P'(\frac{4}{5})=\frac{48}{25}$  \\ 
\hline 
\end{tabular} 
\caption{Table of $\lambda_{ij}$ at Step 3}
\end{table}

\begin{table}[H]
\centering
\begin{tabular}{|c|c|c|}
\hline 
$i \backslash j$ & 2 & 3 \\ 
\hline 
1 & 441/16 & 12  \\ 
\hline 
2  & 625/4 & 144/25 \\ 
\hline 
\end{tabular} 
\caption{Table of $\lambda_{ij}p_{ij}$ at Step 3}
\end{table}

\begin{align*}
w^{\emptyset}_{j} \beta_{\emptyset}
	+w^{\{ 1 \}}_{j} \beta_{\{ 1 \}}
	+w^{\{ 1,4 \}}_{j} \beta_{\{ 1,4\}}&=\min \{ p_{ij}\lambda_{ij}   \} \\
 \beta_{\emptyset}+\beta_{\{ 1 \}}
+\beta_{\{ 1,4\}}&=\min \{ p_{ij}\lambda_{ij}   \}\\
\beta_{\{ 1,4 \}}&=\frac{144}{25}-\frac{3}{4}\\
\beta_{\{ 1,4 \}}&=\frac{501}{100}
\end{align*}

Job 3 is affected to machine 2, $\gamma_3=\frac{144}{25}$ and $T=\{ 1,3,4 \}$.

\begin{figure}[H]
\begin{center}
\includegraphics[width=0.7\textwidth]{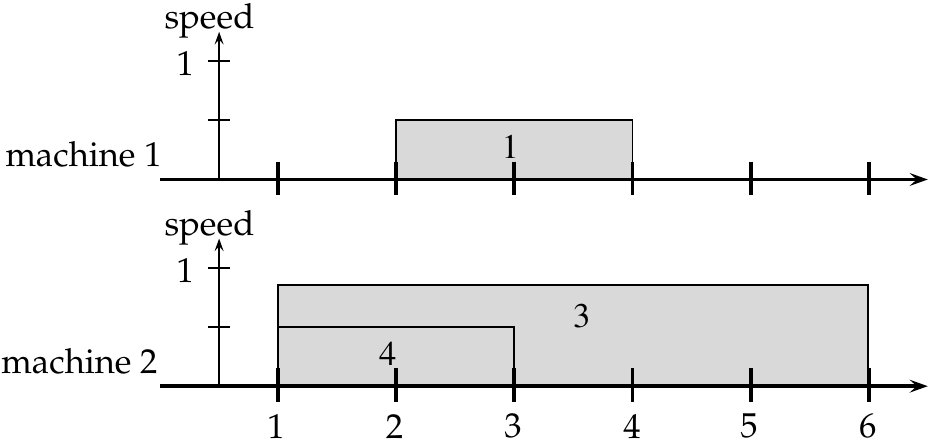} 
\end{center}
\caption{Speed profile $v_i(t)$ at the end of Step 3}
\label{speed_profile_step3}
\end{figure}

\end{document}